\documentclass{llncs}
\pdfoutput=1
\usepackage{graphicx}
\usepackage[scriptsize]{subfigure}
\usepackage{epstopdf}
\usepackage[font=footnotesize,labelfont=bf]{caption}

\usepackage{amsmath}

\RequirePackage[OT1]{fontenc}

\usepackage{bbm,bm}
\usepackage{amsmath,amssymb}
\usepackage{amsfonts}
\usepackage{float}
\usepackage{url}
\usepackage{algorithm}
\usepackage{algpseudocode}
\algrenewcommand\alglinenumber[1]{\tiny #1:}
\usepackage[stable]{footmisc}

\numberwithin{equation}{section}

\newcommand{\E}{\mathbb{E}}
\newcommand{\PP}{\mathbb{P}}

\newcommand{\field}[1]{\mathbbm{#1}}
\newcommand{\R}{\field{R}}

\newcommand{\Ult}{\mathbf{U}^{\leq t}}

\newcommand{\Ulk}{\mathbf{U}^{\leq k}}
\newcommand{\Glt}{\mathbf{G}^{\le t}}

\newcommand{\Flt}{\mathbf{F}^{\le t}}

\begin{document}

\begin{frontmatter}

\title{Probabilistic verification of partially observable dynamical systems\footnote{This is an expanded version of a paper submitted originally in January 2014. Parts of this work are based on Benjamin M. Gyori, \emph{Probabilistic approaches to modeling uncertainty in biological pathway dynamics}, PhD thesis, National University of Singapore, 2014.}}
\titlerunning{Probabilistic verification of partially observable dynamical systems}

\author{Benjamin M. Gyori \and Daniel Paulin \and Sucheendra K. Palaniappan}

\authorrunning{BM. Gyori \and D. Paulin \and SK. Palaniappan}

\author{Benjamin M. Gyori \inst{1}
\and Daniel Paulin\inst{2}
\and Sucheendra K. Palaniappan \inst{3}}
\institute{Department of Systems Biology, Harvard Medical School, USA
\and Department of Statistics and Applied Probability, National University of Singapore, Singapore \and 
INRIA Rennes, France}

\maketitle

\vspace{-0.2in}

\begin{abstract}
The construction and formal verification of dynamical models is important in engineering, biology and other disciplines. We focus on non-linear models containing a set of parameters governing their dynamics. The value of these parameters is often unknown and not directly observable through measurements, which are themselves noisy. When treating parameters as random variables, one can constrain their distribution by conditioning on observations and thereby constructing a posterior probability distribution. We aim to perform model verification with respect to this posterior. The main difficulty in performing verification on a model under the posterior distribution is that in general, it is difficult to obtain \emph{independent} samples from the posterior, especially for non-linear dynamical models. Standard statistical model checking methods require independent realizations of the system and are therefore not applicable in this context. 

We propose a Markov chain Monte Carlo based statistical model checking framework, which produces a sequence of dependent random realizations of the model dynamics over the parameter posterior. Using this sequence of samples, we use statistical hypothesis tests to verify whether the model satisfies a bounded temporal logic property with a certain probability. We use sample size bounds tailored to the setting of dependent samples for fixed sample size and sequential tests. We apply our method to a case-study from the domain of systems biology, to a model of the JAK-STAT biochemical pathway. The pathway is modeled as a system of non-linear ODEs containing a set of unknown parameters. Noisy, indirect observations of the system state are available from an experiment. The results show that the proposed method enables probabilistic verification with respect to the parameter posterior with specified error bounds.
\end{abstract}

\end{frontmatter}

\section{Introduction}
Dynamical systems are used to model the evolution of a system's state in time, and are widely used in science and engineering. The parameters describing the dynamics of these systems are often unknown and one has to condition on noisy data to infer their values. It is then of interest to perform probabilistic verification on a model under this form of \emph{posterior} uncertainty. 
This verification problem has not been addressed in the context of non-linear dynamical systems.  The approximate probabilistic verification of dynamical systems usually involves simulating independent realizations of the system dynamics. However, when conditioning on observations through a posterior, obtaining independent realizations will not be possible, except in some very restricted cases. In this paper we propose a novel method to perform probabilistic verification approximately (but with statistical guarantees) in this context. 

Properties about dynamical systems can be formally expressed as formulas in temporal logic. Using temporal logic, one can conveniently describe both qualitative and quantitative dynamical properties of interest. The technique of \emph{model checking}~\cite{clarkebook} is used to automatically verify if a model satisfies these properties. Model checking has been used for the analysis of dynamical systems in domains including embedded systems~\cite{alur1996automatic} and systems biology~\cite{fisher2007predictive,statmodelcheckTcell,Heath2008}. Both temporal logic and model checking techniques have also been extended to analyze dynamical systems with a component of stochasticity. In this context one aims to verify if a property is satisfied with a certain probability.

Statistical probabilistic verification aims to check whether a dynamical system $\mathcal{S}$ satisfies a temporal logic property $\psi$ with probability at least $r$, or more formally, whether $\mathcal{S} \models \PP_{\ge r}(\psi)$. For a particular realization (also called a \emph{trajectory}) of the system, $\psi$ is either satisfied or not. By imposing a probability measure over the set of trajectories, one can define the probability of satisfaction of $\psi$, denoted $P_{\psi}$. This probability is compared to a threshold $r$, and the verification problem can be posed as a hypothesis test between $H_0: P_{\psi} \ge r+\delta$ and $H_1: P_{\psi} \le r-\delta$, with $\delta$ being a chosen indifference region~\cite{younes06}. The hypothesis test is usually solved using statistical approximations based on repeated simulation of the system \cite{legay2010statistical}.

Here we focus on dynamical systems modeled as a set of coupled ordinary differential equations (ODEs). The dynamics of ODEs is governed by a set of kinetic parameters whose value is often not known and not directly observable. The uncertainty in the parameter values can be represented conveniently by a probability distribution. As the parameter values determine the system dynamics, the parameter distribution also induces a probability measure over the possible realizations of the system dynamics. 

In practice, one often has access to a set of \emph{observations} about a dynamical system's state in time. For instance, when modeling the dynamics of biochemical pathways, one can usually measure the concentration of some molecular species at a few discrete time points. Partial observability arises when the full state of the underlying system cannot be exactly determined through observations. The existence of measurement noise, or the fact that not all components of the system can be measured thus results in partial observability. By conceptually treating model parameters as part of the state, the fact that parameter values cannot be directly measured also implies partial observability. While observations will not reveal the exact value of parameters, one can \emph{condition on} observations to constrain the distribution of the parameters. Following the Bayesian terminology we refer to the probability distribution of parameters conditioned on a set of observations as the \emph{posterior} distribution \cite{box1973bayesian}. 

Sampling independent system trajectories according to a prior distribution is usually straightforward. This allows statistical model checking subject to prior uncertainty, as in \cite{palaniappan2013}. However, obtaining \emph{independent} samples from a posterior distribution is challenging in all but very special cases. The posterior is proportional to the product of the likelihood of the observations and  the prior probability of the parameters. However, evaluating the normalizing constant is not feasible in practice, and sampling independently from the posterior is not possible in general. Approximate probabilistic verification relying on independent samples cannot be used in this setting, and we are not aware of any previous work that has addressed this limitation. 

Here we develop a methodology for the probabilistic verification of a model defined in terms of such a posterior distribution. The method relies on taking a sequence of \emph{dependent} samples from the set on which the posterior distribution is defined (in our case, this is the set of model parameters) using a Markov chain. The Markov chain is designed so that the sequence of states of the chain are samples from the posterior distribution. This method is called Markov chain Monte Carlo (MCMC) \cite{GilksMCMC}. The sequence of samples obtained using MCMC can be used to generate realizations of the system dynamics and to calculate the empirical ratio of realizations for which a temporal logic property is satisfied. However, since these samples are not independent, the standard analysis used to bound the errors on the performed hypothesis tests (as in \cite{herault2004approximate,younes06}) is no longer applicable. 

We rely on recent results in hypothesis testing to bound the number of samples needed to do statistical model checking when one has to rely on dependent samples collected using MCMC \cite{gyori2014hypothesis}. We use these error bounds for the case when one performs the hypothesis test based on a \emph{fixed sample size} as well as the case of \emph{sequential hypothesis testing}, where sample size is not fixed in advance. These tests are similar in nature to ones used in case of independent samples (see \cite{herault2004approximate,younes06}), but are tailored to the case of posterior verification. 

Probabilistic verification on dynamical models of biochemical pathways (including ODE models, as well as discrete or continuous time Markov chains) is an active and increasingly important field \cite{brim2013model}. In pathway models, the value of relevant kinetic parameters is rarely known, and inferring parameters based on observations is an important and difficult problem. Due to the richer analysis it enables, Bayesian inference is increasingly adopted for treating model parameters as random variables and making predictions with respect to their posterior distribution \cite{learningsysbbook,eydgahi2013properties}. However, as of now, probabilistic verification with respect to posterior parameter distributions has not been demonstrated in the context of pathway models. 
We illustrate our method on an ODE model of the JAK-STAT biochemical pathway. The system of ODEs describing the pathway contain parameters whose values are unknown, and noisy and indirect measurements of the system state are available from biological experiments. We are interested in formally verifying the dynamical properties of this system with respect to bounded temporal logic properties. We show that using our method, probabilistic verification is possible with specified error bounds. 

In the next section, we introduce ODE models, their dynamics, and the notion of partial observability. In Section 3, we introduce our temporal logic specification for expressing dynamical properties. In Section 4, we describe our main algorithms for performing statistical model checking. In Section 5 we apply our method on a case study from the domain of systems biology. Section 6 concludes our paper with possible extensions in the future. 


\section{ODE models and partial observability}\label{sec:ode}
Systems of ODEs are commonly used for modeling dynamics in a wide variety of disciplines, including systems biology \cite{aldridge2006physicochemical,klipp2005systems}. An ODE system describes the time-derivative of a set of variables $\mathbf{x}(t) \in \R^{d_x}$ through a system of (possibly non-linear) equations. We also allow for a set of input variables $\mathbf{u}(t) \in R^{d_u}$, and explicitly include a model for observing the state of the system through output variables $\mathbf{y}(t)  \in R^{d_y}$. The equations are stated as follows. 
\begin{align}\label{dynsysmodel}
\dot{\mathbf{x}}(t) &= f(\mathbf{x}(t),\mathbf{u}(t),\theta) \nonumber \\
\mathbf{y}(t) &= g(\mathbf{x}(t)) + \mathbf{w}(t).
\end{align}
Here $\theta \in \R^{d_{\theta}}$ is a vector of model parameters and $\mathbf{w} \in R^{d_y}$ denotes the noise component of observations. We assume that the form of the functions $f,g$ and the probability distribution of $\mathbf{w}$ are known. To simulate the model, initial conditions $\mathbf{x}(0)$ need to be set, and throughout the rest of the paper we assume that these are given. However, if this is not the case, initial conditions could also be treated as part of the set of unknown model parameters (as, for instance, in \cite{vanlier2012integrated}).  

The notion of \emph{partial observability} expresses that we do not have direct access to $\mathbf{x}(t)$, and can only observe the state indirectly through $\mathbf{y}(t)$. Observations only provide partial information about the underlying system for any of the following reasons: (i) observations are noisy (ii) not all state variables can be observed (iii) observations do not map uniquely to specific states. The concept of partial observability can also be extended to the set of parameters $\theta$. In this case $\mathbf{y}(t)$ provides indirect information on $\theta$ only through observing $\mathbf{x}(t)$.

We constrain model parameters (whose values are not exactly known) to be in a set $\Theta \subset \R^{d_{\theta}}$, and for simplicity define this set as the hypercube arising by constraining parameter $\theta_i$ to the interval $[a_i,b_i]$, where $a_i < b_i \in \R$, $1 \le i \le d_{\theta}$. The set of possible parameter values will thus be $\Theta = [a_1,b_1] \times [a_2,b_2] \times \ldots \times [a_{d_{\theta}},b_{d_{\theta}}]$.

Importantly, we assume that a prior probability density $p_0(\theta)$ is given over $\Theta$. One can use the prior to encode existing knowledge about the joint distribution of parameters. In the simplest case, $p_0(\theta)$ will be uniform over $\Theta$, defined as
\begin{equation}
p_0(\theta) = 
\begin{cases}
c & \text{if } \theta \in \Theta \\
0 & \text{otherwise.}
\end{cases}
\end{equation}
Here $c = \left( \prod_{i=1}^{d_{\theta}} (b_i - a_i)\right)^{-1}$ is a constant ensuring that $p_0$ integrates to $1$.

Now assume that we have a set of observations $Y$ obtained by gathering instances of the output $\mathbf{y}$. The set of observations contains vector values of $\mathbf{y}(t)$ at a finite, discrete set of time steps: $Y = \{\mathbf{y}(t_1), \mathbf{y}(t_2), \ldots \mathbf{y}(t_{\ell})\}$, and we denote by $Y_{i,j}$ the $i$th component of the vector $\mathbf{y}(t_j)$. 
The observation process $\mathbf{y}(t)$ is indirectly dependent on the model parameters, and therefore the set of measurements $Y$ contains indirect information on the value of parameters. 

Our goal is to constrain the probability distribution over the model parameters, and construct a posterior distribution by conditioning on the set of observations \cite{box1973bayesian}. We denote the posterior distribution of the parameters $\pi(\theta|Y)$, which, by the Bayes theorem can be expanded to
\begin{equation}\label{eqposterior}
\pi(\theta|Y) = \frac{p(Y|\theta)p_0(\theta)}{p(Y)} = \frac{p(Y|\theta)p_0(\theta)}{\int_{\Theta}p(Y|\theta)p_0(\theta)d\theta}.
\end{equation}
In the above equation $p(Y|\theta)$ is the probability of an observation conditioned on $\theta$. However, since $Y$ is fixed throughout the analysis, $p(Y|\theta)$ is considered a function of $\theta$, and it is commonly referred to as the \emph{likelihood}. The form of the likelihood function is known due to the fact that the noise component $\mathbf{w}$ is of a known distribution. In many applications, $\mathbf{w}$ is a vector of $d_{y}$ independent Gaussian random variables. In this special case, for a particular $\theta^n$ we have
\begin{equation}\label{eq:guasslh}
p(Y|\theta^n) = \prod_{i=1}^{d_y}\prod_{j=1}^{\ell} P(Y_{i,j}|\theta^n) = C \exp \left(-\sum_{i=1}^{d_y}\sum_{j=1}^{\ell}\left(\frac{Y_{i,j}-y_{i}(t_j)|_{\theta^n}}{\sqrt{2}\sigma_{i,j}}\right)^2\right),
\end{equation}
where $y_{i}(t_j)|_{\theta^n}$ denotes the $i$th component of the output of the model when using parameters $\theta^n$, $\sigma_{i,j}$ is the standard deviation of data point $Y_{i,j}$, and $C$ is a normalization constant. 

Given a particular set of parameters, evaluating the prior is straightforward. Calculating the likelihood requires simulating the system up to the time point $t_{\ell}$, and evaluating the obtained trajectories against the measurement data. One can use the same concept if observations are given for multiple measurement conditions by simulating for each condition to evaluate the joint likelihood. The main difficulty in dealing with posteriors is posed by the factor $p(Y)$ in \eqref{eqposterior}, which is usually intractable to evaluate in practice. The fact that the form of the posterior is hard to represent essentially prevents the use of independent samples from the posterior. The proposed MCMC method (described in Section \ref{sec:smcmcmc}), provides a sequence of dependent samples from the posterior in a way that the factor $p(Y)$ need not be evaluated. 

In the next section we introduce the temporal logic used to formalize properties on the realizations of the dynamical system.

\section{Expressing dynamical properties using PBLTL} 
To specify the dynamical properties of a single realization of the system, we first encode them as formulas in a specification logic. We assume that we are concerned with analyzing the dynamics of the system only up to a maximal time point $\tau$. We use a bounded version of linear time temporal logic (BLTL)\cite{clarkebook} for this. The formulas in this logic would be interpreted at a finite set of time points $\mathcal{T} = \{0, 1, \ldots, \tau\}$ corresponding to all the relevant time points of interest. 

In our setting  a trajectory is represented by $\varsigma_{\theta}$,  which (given fix initial conditions) is fully defined by the choice of parameters ${\theta}$ since the ODE system is deterministic. A trajectory will be defined by the set of states $\varsigma_{\theta}$ = ($\mathbf{x}(0)|_{\theta}$, $\mathbf{x}(1)|_{\theta}$,$\ldots$ , $\mathbf{x}(\tau)|_{\theta}$), where $\mathbf{x}(i)|_{\theta}$ is the value of system variables at time point $i$ when the corresponding ODEs are simulated with the parameter set $\theta$. $\varsigma_{\theta}(t) = \mathbf{x}(t)|_{\theta} $ for $t \in \mathcal{T}$. The transitions from  $\mathbf{x}(i)|_{\theta}$ to  $\mathbf{x}(i+1)|_{\theta}$ is ensured by the fact that once we fix the parameters values, the systems of ODEs has a unique solution and is characterized by a continuous function (for more details, see \cite{palaniappan2013}).

Atomic propositions in BLTL will be of the form $(i,L,U)$ with $L$ $\le$ $U$. This will be interpreted as ``the value of $x_i$ falls in the interval $[L, U]$". In Section \ref{sec:results}, for easier readability, we will use the $[L \le x_i \le U]$ notation with the same intended meaning. 

The syntax of formulas in BLTL are defined in a standard way: 
(i) Every atomic proposition is a BLTL formula.
(ii) The constants $\emph{true}$, $\emph{false}$ are BLTL formulas.
(iii) If $\psi$, $\psi'$ are BLTL formulas then $\lnot \psi$ and $\psi \vee \psi'$  are BLTL formulas. 
(iv) If $\psi$, $\psi'$ are BLTL formulas  then $\psi \Ult  \psi'$ is a  BLTL formula, where $t \leq \tau$ is a positive integer.

Derived operators such as $\wedge$, $\supset$, $\equiv$, $\Glt$, and  $\Flt$ are defined in the usual way. Qualitative properties of the system dynamics defined by the ODEs can be efficiently expressed using BLTL (see for instance \cite{palaniappan2013}).

The semantics of BLTL will be defined by $\varsigma_{\theta}, t \models \psi$ as follows.

\begin{itemize}
\item $\varsigma_{\theta}, t \models (i,L,U)$ iff $L \leq \varsigma_{\theta,i}(t) \leq U$ where $\varsigma_{\theta,i}(t)$ is the $i$th component of  $\varsigma_{\theta}(t)$.
\item $\varsigma_{\theta}, t \models \psi \vee \psi' $ iff  $\varsigma_{\theta}, t \models \psi $ or $\varsigma_{\theta}, t \models \psi' $.
\item $\varsigma_{\theta}, t \models \lnot \psi $ iff  $\varsigma_{\theta}, t \not\models \psi $.
\item $\varsigma_{\theta}, t \models \psi \Ulk \psi'$ iff there exists $k'$ such that $k' \leq k$,\,\, $t + k' \leq \tau$, $\varsigma, t + k' \models \psi'$ and  $\varsigma_{\theta}, t + k'' \models \psi$ for every $0 \leq k'' < k'$.
\end{itemize}

Under assumptions of continuity and measurability on the ODE equations, we can assign a probability to the trajectories satisfying a given formula $\psi$ with respect to the distribution of parameters (for a proof of the fact that this probability exists, see \cite{palaniappan2013}). We now define the probability of the system satisfying a formula $\psi$ as
\begin{equation}
P_{\psi} = \int_{\Theta} \pi(\theta|Y) I(\varsigma_\theta \models \psi) d\theta,
\end{equation}
where $I$ is the indicator function taking value $1$ if $\varsigma_\theta \models \psi$, and $0$ otherwise. 

To express properties of this nature, we will encode them in a formalism called PBLTL\cite{jha2009statistical}, which is a probabilistic extension of BLTL. Formulas in PBLTL are of the form  $\PP_{\ge r}(\psi) $ (or $\PP_{\le r}(\psi) $ ), where $\psi$ is a BLTL formula and $r$ is a real number in $(0,1)$. The PBLTL formula $\PP_{\ge r}(\psi)$ expresses that we want to verify whether the probability measure of the trajectories satisfying $\psi$ (or $P_{\psi}$) is at least $r$. The next section introduces our statistical framework for deciding approximately, but with statistical guarantees, whether the model satisfies properties expressed in PBLTL.

\section{Statistical model checking using MCMC}
\label{sec:smcmcmc}
In this section we develop the methodology for performing statistical model checking with respect to the Bayesian posterior distribution $\pi(\theta|Y)$.

Our goal is decide between the following two hypotheses.
\begin{align}\label{eqhypothesis}
H_0: \quad & P_{\psi} \ge r + \delta,  \\
H_1: \quad & P_{\psi} \le r - \delta, \nonumber
\end{align}
where $\PP_{\ge r}(\psi)$ is a PBLTL formula, $r \in (0,1)$ and $\delta \in (0,\min(r,1-r))$. 

\subsection{Markov chain construction}

In Section \ref{sec:ode} we have discussed that independent realizations of the system with respect to the posterior distribution cannot be obtained. We will therefore use a sequence of dependent samples from a Markov chain to decide between the hypotheses. 
We define a Markov chain whose state space is the space of parameters. The chain starts at an initial parameter sampled from the \emph{prior}. In each subsequent step of the chain, one first uses a \emph{proposal distribution} to pick the next candidate parameter, and then applies the \emph{acceptance ratio} to accept or reject the proposed candidate. At each step of the Markov chain, the trajectory corresponding to the current parameter values is verified, and these samples are used to perform probabilistic verification with respect to the posterior. The key idea is to design the Markov chain in a way that its stationary distribution matches the posterior $\pi(\theta|Y)$.

There are many possible ways to construct an adequate proposal distribution. We denote the proposal by $q(\theta^n \to \theta)$, which represents the probability of proposing $\theta$ if the current parameter value is $\theta^n$. We suggest using $q(\theta^n \to \theta) = \mathcal{N}(\theta^n,\Sigma_{\mathrm{MH}})$, a $d_{\theta}$-dimensional multivariate Gaussian with mean identical to the current parameter vector, and covariance matrix $\Sigma_{\mathrm{MH}}$. Here $\Sigma_{\mathrm{MH}}$ can be diagonal with entries $\sigma_{\mathrm{MH},1}^2, \ldots, \sigma_{\mathrm{MH},d_{\theta}}^2$, representing variances along each dimension independently. 
In practice it is important to choose the entries of the covariance matrix carefully, since it greatly affects the mixing properties of the chain. (For more details on constructing efficient proposal steps, such as adaptive schemes, we refer the reader to \cite{RobertsCasella}).
The acceptance ratio follows from the Metropolis-Hasting scheme, where the candidate is accepted with probability $\alpha$, in general, determined by the proposal and the posterior as follows.
\begin{equation}
\alpha = \min \left (1, \frac{q(\theta' \to \theta^n)}{q(\theta^n \to \theta')} \frac{\pi(\theta'|Y)}{\pi(\theta^n|Y)} \right) = \min \left (1, \frac{q(\theta' \to \theta^n)}{q(\theta^n \to \theta')} \frac{p_0(\theta')p(Y|\theta')}{p_0(\theta^n)p(Y|\theta^n)} \right ).
\end{equation} 
Note that the normalization constant ($p(Y)$) appearing in the posterior is eliminated, and one thus needs only evaluate the prior and the likelihood at the original and at the proposed parameter value.

The proposal and acceptance steps defined as above form an instance of the Metropolis-Hasting algorithm, which is proven to converge to the desired target distribution~\cite{GilksMCMC}. In practice, one takes an initial $t_0$ number of steps in the Markov chain, called the ``burn-in time", to ensure that the chain has sufficiently converged to the posterior. 

\subsection{Hypothesis tests}

We introduce the function getMCMCsample, which takes as input the current parameter values, takes a single step in the Markov chain, and returns the new parameter values. 

\begin{algorithm}[h]
\renewcommand{\thealgorithm}{}
\makeatletter\renewcommand{\ALG@name}{Function}
\caption{getMCMCsample}
\small
Input: parameter vector $\theta_{\mathrm{in}}$.
Output: parameter vector $\theta_{\mathrm{out}}$ 
\begin{algorithmic}[1]
   	\State Sample a new parameter vector based on proposal: $\theta' \sim q(\theta_{\mathrm{in}} \to \theta)$
   	
   	\State Calculate acceptance ratio $\alpha = \min \left (1, \frac{p_0(\theta')p(Y|\theta')q(\theta' \to \theta_{\mathrm{in}})}{p_0(\theta_{\mathrm{in}})p(Y|\theta_{\mathrm{in}})q(\theta_{\mathrm{in}} \to \theta')} \right )$
   	
   	\State Generate $\eta \sim \text{Uniform}[0,1]$
   	
   	\If {$\eta < \alpha$} 
   		\State \textbf{return} $\theta_{\mathrm{out}}:=\theta'$
   	\Else 
   		\State \textbf{return} $\theta_{\mathrm{out}}:=\theta_{\mathrm{in}}$
   	\EndIf
\end{algorithmic}
\end{algorithm}

We present two tests between the hypotheses in \eqref{eqhypothesis}. These tests use getMCMCsample as a subroutine. The first test assumes that we have fixed $N$, the total number of samples to collect, and thus a choice of either $H_0$ or $H_1$ is returned after exactly $N$ steps. 

\begin{algorithm}[h]
\renewcommand{\thealgorithm}{1}
\caption{Fixed sample size hypothesis test}
\small
\label{algfix}
Input: BLTL property $\psi$, threshold probability $r$, observations $Y$, number of samples $N$, number of burn-in steps $t_0$, prior $p_0$, proposal $q$.

Output: Choice of $H_0$ or $H_1$.
	\begin{algorithmic}[1]
		\State Sample initial parameter vector from the prior $\vartheta^0 \sim p_0(\theta)$
		\For{$i:=1 \ldots t_0$}
			\State $\vartheta^{i} :=$ getMCMCsample($\vartheta^{i-1}$)
		\EndFor
		
		\State Set $S := 0$ and $\theta^0 := \vartheta^{t_0}$
		\For{$n:=1 \ldots N$}
			\State $\theta^{n} :=$ getMCMCsample($\theta^{n-1}$)
			\State Simulate the trajectory $\varsigma_{\theta^n}$
			\If{$\varsigma_{\theta^n} \models \psi$}
				\State $S := S + 1$
			\EndIf
		\EndFor	
		
		\If{$S \ge Nr$}
			\State \textbf{return} $H_0$
		\Else
			\State \textbf{return} $H_1$
		\EndIf
		\item[]
		\item[]
	\end{algorithmic}
\end{algorithm}

The second test uses sequential hypothesis testing to adaptively set the number of steps before stopping (based on the result of verification on samples gathered so far). The stopping condition is governed by a threshold $M$. The value of $N$ (or $M$ respectively) is chosen depending on $r$, $\delta$ and the required Type-I and Type-II error limit $\epsilon$. In fact, $\delta$ and $\epsilon$ do not explicitly appear in the algorithms, and only influence it through the chosen value of $N$ or $M$. We now discuss how to choose $N$ and $M$ to obtain a test with error bound $\epsilon$.

\begin{algorithm}[H]
\renewcommand{\thealgorithm}{2}
\caption{Sequential hypothesis test}
\small
Input: BLTL property $\psi$, threshold probability $r$, observations $Y$, stopping condition $M$, number of burn-in steps $t_0$, prior $p_0$, proposal $q$.

Output: Choice of $H_0$ or $H_1$.
	\begin{algorithmic}[1]
		\State Sample initial parameter vector from the prior $\vartheta^0 \sim p_0(\theta)$
		\For{$i:=1 \ldots t_0$}
			\State $\vartheta^{i} :=$ getMCMCsample($\vartheta^{i-1}$)
		\EndFor
		
		\State Set $n:=1$, $S := 0$ and $\theta^0 := \vartheta^{t_0}$
		\Loop
			\State $\theta^{n} :=$ getMCMCsample($\theta^{n-1}$).
			\State Simulate the trajectory $\varsigma_{\theta^n}$
			\If{$\varsigma_{\theta^n} \models \psi$}
				\State $S := S + 1$
			\EndIf
			
			\If{$S \ge nr + M$}
				\State \textbf{return} $H_0$
			\ElsIf{$S \le nr - M$}
				\State \textbf{return} $H_1$
			\Else
				\State Set $n:=n+1$ and continue
			\EndIf
		\EndLoop

	\end{algorithmic}
\end{algorithm}

\subsection{Choosing the sample size}\label{sec:samplesize}
The statistical theory behind performing hypothesis tests on samples obtained from a Markov chain was developed in \cite{gyori2014hypothesis}. 
In \cite{gyori2014hypothesis}, concentration inequalities are used to bound the absolute difference between the empirical average $1/n \sum^{n}_{i=1}f(\theta^i)$, and the true (unknown) expected value $\E_{\pi}f$, for a function $f:\Theta\to\R$. The conditions for the inequalities to hold are that $\theta^i$ are states of a reversible Markov chain whose stationary distribution is $\pi$, further, it is required that $f$ is square integrable ($f\in L^2(\pi)$), and $0 \le f \le 1$.  
These conditions are satisfied in our setting with $f$ corresponding to the outcome of verification as
\begin{align}
f(\theta) = 
\begin{cases}
1 & \text{if } \varsigma_{\theta} \models \psi, \\
0 & \text{otherwise. }
\end{cases}
\end{align}

A key parameter appearing in the concentration inequalities is the \emph{spectral gap} of the Markov chain, which we denote $\gamma$. The spectral gap is a measure of the speed of mixing of the chain, which needs to be estimated in practice. Now we briefly review the iterative method for estimating the spectral gap, as introduced in the Appendix of \cite{gyori2014hypothesis}.
As before, we assume that $\theta=(\theta_1,\theta_2,\ldots,\theta_{d_{\theta}})\in \R^{d_{\theta}}$.
\begin{enumerate}
\item Run an initial simulation  of length $n$ yielding parameter values $\theta^1,\ldots, \theta^{n}$. In every step $1\le i\le n$, save each component $\theta^{i}_1,\ldots, \theta^{i}_{d_{\theta}}$.
\item Set $\eta=1$, and for each $1\le k\le d_{\theta}$, compute
\begin{equation}\label{eq:gammahat}
\hat{\gamma}_{\eta, k} := 1-(\hat{\rho}_{\eta,k}/\hat{V}_{k})^{1/\eta},
\end{equation}
where
\begin{align}
\hat{V}_k&:= \frac{1}{n} \sum_{i=1}^{n}\theta^i_{k} - \left(\frac{1}{n} \sum_{i=1}^n \theta^i_{k}\right)^2 \\
\hat{\rho}_{\eta, k}(f) &:= \frac{1}{n-\eta}\sum_{i=1}^{n-\eta}\left(\theta^i_{k}-
\frac{1}{n-\eta}\sum_{j=1}^{n-\eta}\theta^j_{k}
\right)\left(\theta^{i+\eta}_{k}-
\frac{1}{n-\eta}\sum_{j=1}^{n-\eta}\theta^{j+\eta}_{k}
\right).
\end{align}

Denote the minimum of $\hat{\gamma}_{\eta, 1},\ldots, \hat{\gamma}_{\eta, d_{\theta}}$ by $\hat{\gamma}_{\min}(1)$, and compute 
\begin{equation}\label{eq:etadef}\eta(1):=\frac{\log(n\hat{\gamma}_{\min}(1))}{4\log(1/(1-\hat{\gamma}_{\min}(1)))}.\end{equation}
\item Inductively assume we have already computed $\eta(j)$  for $j\ge 1$ (based on \eqref{eq:etadef}). Then compute $\gamma_{\min}(j+1)$ based on \eqref{eq:gammahat} using $\eta=\eta(j)$. If $\hat{\gamma}_{\min}(j+1)\ge \hat{\gamma}_{\min}(j)$, then stop, and let $\hat{\gamma}:=\hat{\gamma}_{\min}(j)$. Otherwise compute $\eta(j+1)$ and repeat this step.
\item To ensure a sufficient amount of initial data, if $n$ satisfies $n>100/\hat{\gamma}$, accept the estimate, otherwise choose $n=200/\hat{\gamma}$ and restart from Step 2.
\end{enumerate}

We now give results based on \cite{gyori2014hypothesis} to choose the needed sample size $N$ and stopping condition $M$ for the fixed sample size and sequential test, respectively, with which verification with error probability at most $\epsilon$ is achieved.

\begin{proposition}
The probability of choosing the incorrect hypothesis in Algorithm 1 is at most $\epsilon$ with the choice of 
\begin{equation}
N \ge \frac{\log(1/\epsilon)}{\gamma\delta^2}.
\end{equation}
\end{proposition}
\begin{proof}
Proposition 3.2 of \cite{gyori2014hypothesis} proves that the probability of choosing the incorrect hypothesis in a fixed sample size test is bounded by 
\begin{equation}\label{eq:nbound}
\exp(-\gamma\delta^2n).
\end{equation}
The proposition follows by rearrangement. 
\end{proof}

\begin{proposition}
The probability of choosing the incorrect hypothesis in Algorithm 2 is at most $\epsilon$ with the choice of
\begin{equation}\label{eq:seqtestM}
M = \frac{\log(2/(\epsilon\gamma\delta^2))}{2\gamma\delta+\gamma\delta^2/(1-r)}.
\end{equation}
\end{proposition}
\begin{proof}
Proposition 3.3 of \cite{gyori2014hypothesis} proves that the probability of choosing the incorrect hypothesis in a sequential test is bounded by 
\begin{equation}\label{eq:seqtesterror}
\exp(-2\gamma\delta M)\cdot \exp(-M\gamma\delta^2/(1-r)).
\end{equation}
The proposition follows by rearrangement. 
\end{proof}

\subsection{Decoupling sampling and model checking}\label{sec:decouple}
Typically, one will be interested in verifying several different properties of a model. It is impractical to re-run the full MCMC procedure for each property independently. We can exploit the fact that the Markov chain based sample collection is independent from the model checking task. The sequence of parameter samples collected by the Markov chain only depends on the model and the experimental data, and not on the property that is being verified. 
In practice, it is better to first run the Markov chain for a large number of steps ``off-line", and store the collected parameter samples for later use in verification. 

Assuming that a sufficiently long sequence of parameter samples has been stored, it is possible to run the fix sample size or sequential hypothesis tests on this stored set of samples. In fact, there are two important optimizations that this enables in practice. 

First, many of the parameters that the Markov chain generates are identical. This is because each time a proposed parameter is rejected, the previous parameter is kept (the Markov chain stays in its original state). Depending on the design, the Markov chain will typically have an acceptance rate between $10-40\%$. Naturally, it is enough to perform verification with each distinct parameter, and take into account the multiplicity of the parameter in the hypothesis test. 

Second, it is possible to parallelize the decoupled verification phase. For the fix sample size test, massive parallelization is possible, since each stored parameter can be verified independently. For the sequential test, it is possible to introduce \emph{batches} of samples that are verified in parallel. After verifying a batch of samples, the stopping condition of the sequential test is checked, and the procedure either stops and makes a decision, or another batch of samples is simulated and verified. 

\section{Results}\label{sec:results}
We implemented the proposed method in C++. Here we present a case study from the domain of systems biology, where dynamical system models (and in particular ODE models) are commonly used to understand the temporal behavior of biochemical components inside cells \cite{klipp2005systems}. We apply our method on a model of the JAK-STAT biochemical pathway. The signaling cascade is initiated by erythropoietin (Epo), which, when bound to a receptor, induces the phosphorylation of STAT protein in the cytoplasm. Phosphoylated STAT dimerizes and enters the nucleus where it alters gene expression. Subsequently the nuclear STAT goes through dissociation and dephosphorylation and is transported back into the cytoplasm (see also \cite{vanlier2012integrated}). The set of ODE equations describing the dynamics are given in the Appendix.

The variables in the model and the 4 model parameters (whose values are not known) cannot be directly measured. However, experimental data for two indirect quantities (total phosphorylated STAT, and total STAT in cytoplasm) has been published in \cite{swameye2003}. We use Gaussian likelihood (see \eqref{eq:guasslh}) when comparing the data to simulated trajectories, and 
assume a uniform prior distribution over a range of possible parameter values. 
The parameter vector of the model is $\theta=(k_1,k_2,k_3,k_4)$. The parameter ranges and the covariance matrix diagonal entries ($\sigma_{\mathrm{MH}}$) used to define the MCMC proposal distribution are provided in Table \ref{table:params}.

\begin{table}
\centering
	\begin{tabular}{c|c|c}
	Parameter 	& Limits 	& $\sigma_{\mathrm{MH}}$ \\ \hline 
	$k_1$ 		& $[0,5]$			& $0.02$ \\
	$k_2$		& $[0,30]$			& $0.5$ \\
	$k_3$		& $[0,1]$			& $0.01$ \\
	$k_4$		& $[0,5]$			& $0.02$
	\end{tabular}
\caption{Parameter ranges and entries in the proposal covariance matrix}
\label{table:params}
\end{table}

We use a Markov chain as introduced in Section \ref{sec:smcmcmc} to collect samples from the space of parameters according to the posterior distribution, while evaluating the corresponding trajectories against properties of interest. 
Figure \ref{fig:density} shows the set of parameters collected by one Markov chain. The high-probability region of the parameter posterior has a complex shape with some parameters being well constrained while others showing large uncertainty. 

\begin{figure}[h]
\includegraphics[width=\textwidth]{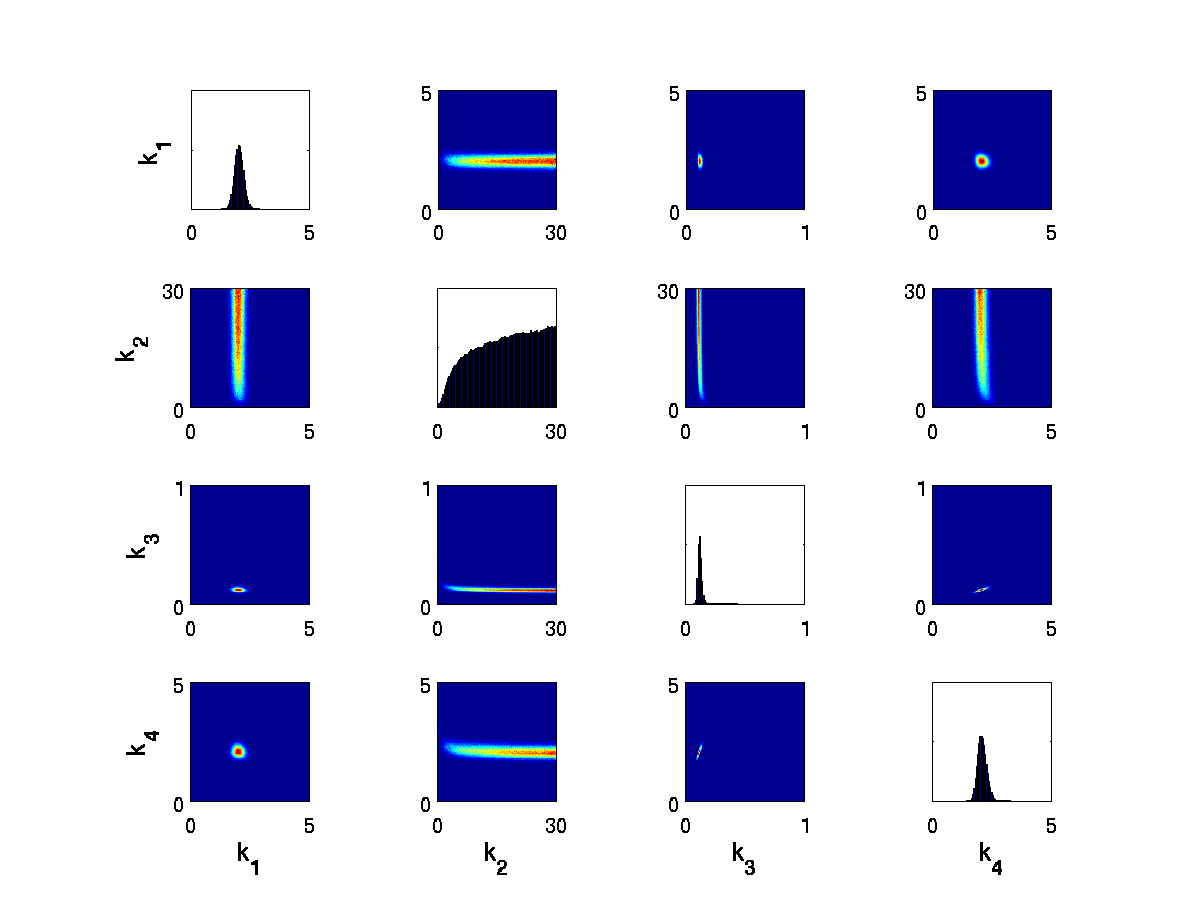}
\caption{Density of samples collected using the Markov chain Monte Carlo approach in the space of model parameters $\theta=(k_1,k_2,k_3,k_4)$. Histograms and 2-dimensional projections of samples are shown. Red color indicates high sample density, and blue indicates low sample density.}
\label{fig:density}
\end{figure}

In performing verification, we are mainly interested in the dynamics of nuclear STAT (STATn), since it is involved in gene expression \cite{swameye2003}. Specifically, we verify dynamical properties of STATn under various types of Epo stimulation (Epo is an input set externally and does not appear in the formulas). We chose $16$ discrete time points between $0$ and $60$ minutes to represent trajectories with respect to BLTL formulas (in the formulas below we will use the absolute time rather than the discrete time index).

\noindent \textbf{Property 1} \emph{$\mathrm{STATn}$ reaches a \emph{high} level (it reaches $1$ but does not cross $1.2$), and then settles at a \emph{low} level under transient Epo stimulation }
\begin{align}
\psi_{1} = &G^{\le 60}[0 \le \mathrm{STATn} \le 1.2] \wedge F^{\le 60}([1 \le \mathrm{STATn} \le 1.2] \nonumber \\
&\wedge F^{\le 60}(G^{\le 60}([0 \le \mathrm{STATn} \le 0.5]))).
\end{align}

\subsection{Method validation}

We use $\PP_{\ge r}(\psi_{1})$ as a case study for validating different aspects of our approach. We ran $m=1000$ independent instances of the MCMC sampler for a total of $2\cdot 10^6$ steps each (with $t_0=5\cdot 10^4$ burn-in steps). To get a reliable estimate of the true underlying probability of satisfaction $P_{\psi_1}$, we took the overall average of the estimates from all $m$ chains, and treated the obtained value $P_{\psi_1} \approx \widehat{P}_{\psi_1}=0.8123$ as the reference for $P_{\psi_1}$. 

Figure \ref{fig:scatter} shows the parameter samples collected by one of the Markov chains according to the satisfaction of $\psi_1$. The projection of samples to the joint space of model parameters $k_1,k_2$ and $k_3$ show the separation between the set of parameters with which $\psi_1$ is satisfied, and ones with which it is not. 
\begin{figure}[h]
\includegraphics[width=\textwidth]{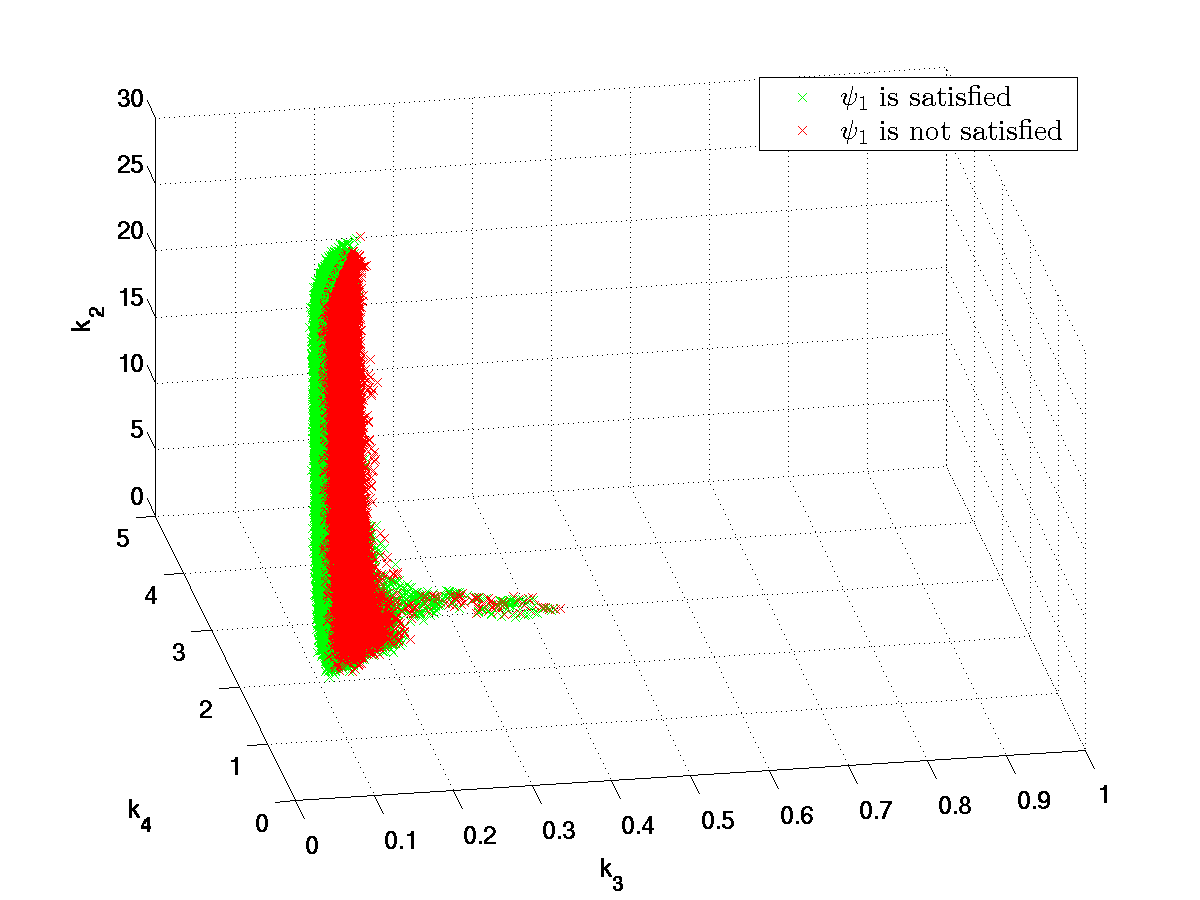}
\caption{Samples collected using the Markov chain Monte Carlo approach in the space of model parameters $\theta=(k_1,k_2,k_3,k_4)$, projected to the space of $(k_2,k_3,k_4)$, with colors indicating the satisfaction of property $\psi_1$ under each parameter combination. }
\label{fig:scatter}
\end{figure}

We use the output of the $m$ independent chains as a basis for constructing results in Figure \ref{fig:results_psi1}(a-d). We used the method described in Section \ref{sec:samplesize} to estimate the value of the spectral gap, for each chain, independently. 

We first examined the empirical error rate of the fixed sample size hypothesis test. We define the empirical error rate $E_n$ as the ratio of chains choosing $H_0$ if $H_1$ holds (or the ratio choosing $H_1$ if $H_0$ holds) after $n$ steps. If neither $H_0$ nor $H_1$ holds (when $r-\delta < P_{\psi} < r + \delta$), then $E_n := 0$.
We set $r=\widehat{P}_{\psi_1}-\delta$ and calculated $E_n$ for a range of sample sizes up to $n=10^6$. For the same set of sample sizes, we calculated the mean error rate bound derived from equation \eqref{eq:nbound} as $\epsilon_n = \exp(-n\gamma\delta)$ (here the mean is used since $\gamma$ is estimated for each chain independently). Figure \ref{fig:results_psi1}(a) shows $E_n$ and $\epsilon_n$ as a function of $n$ for different values of $\delta$. It is apparent that $E_n$ decreases monotonically with increasing sample size $n$, and that $E_n$ is higher for lower values of $\delta$. As seen in Figure \ref{fig:results_psi1}(a), the empirical error rates are consistently below the upper bound ($E_n \le \epsilon_n$ for all examined $n,\delta$).

We next look at results for sequential hypothesis testing. We refer to the number of samples collected in the Markov chain before a decision is made as the \emph{stopping time}.
In Figure \ref{fig:results_psi1} (b), the empirical cumulative distribution of stopping times is shown for the hypothesis test on $P_{\psi_1}$ for a set of $r$ values in $(0,1)$. Here the value of $\delta=0.05$ and $\epsilon=0.01$ is fixed. The distribution of fixed sample sizes for the same hypothesis test is also show as a reference.
The plot shows that for values of $r$ distant from the true probability, sequential sampling consistently terminates with small variability at low sample sizes. When $r$ is close to the true probability, the stopping times show higher variability. 
Figures \ref{fig:results_psi1}(c-d) show the mean empirical stopping times for a range of $r$ values for different values of $\delta$ (c), and different values of $\epsilon$ (d). For values of $r$ close to $\widehat{P}_{\psi_1}$, some chains did not stop within $2\cdot 10^6$ samples, and the corresponding mean values are therefore not determined. These empirical results are consistent with sequential hypothesis testing in the independent sample setting \cite{younes2002probabilistic}.

Finally, we evaluated the empirical error rate in the sequential hypothesis test, and found that out of the $m=1000$ independent runs, no error was made under all examined choices of $r,\epsilon,\delta$. This shows that the specified error bound \eqref{eq:seqtesterror} was indeed met. This also suggests that the bound \eqref{eq:seqtesterror} might not be sharp and $M$ could be chosen even smaller than described by \eqref{eq:seqtestM}, resulting in earlier stopping. 

\begin{figure}[h]
\centering
	\begin{tabular}{ccr}
	\addtolength{\subfigcapskip}{0.2cm}
	
	\subfigure[Empirical error rates for the fixed sample size test for a range of sample sizes. Dashed lines show the theoretical upper bounds derived from \eqref{eq:nbound}. Here $r = \widehat{P}_{\psi_1} - \delta$ and $\epsilon=0.01$ are fixed, and $3$ distinct $\delta$ values are shown.]{
		\includegraphics[height=4.3cm]{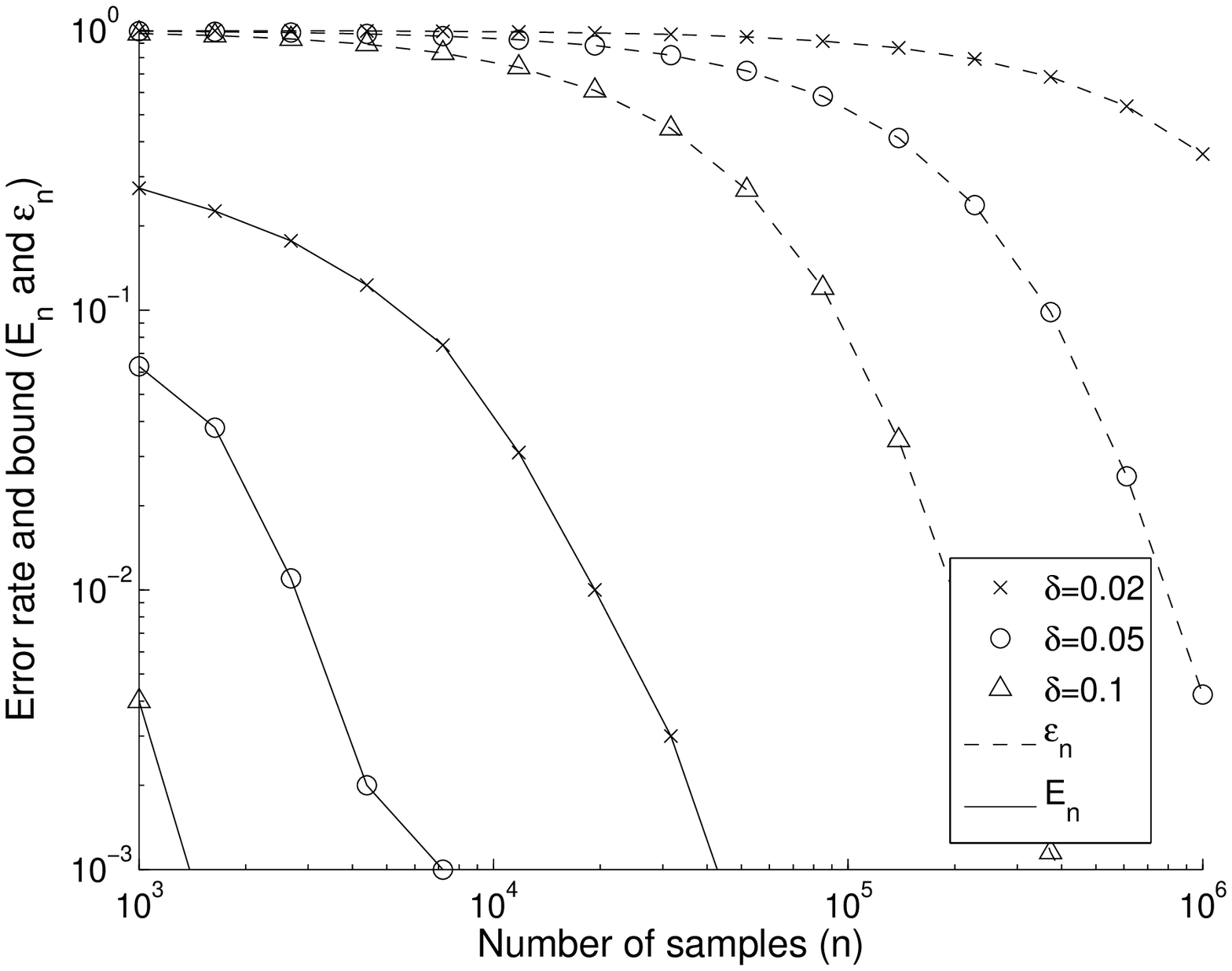}
		}
	&
	\addtolength{\subfigcapskip}{0.2cm}
	\subfigure[Empirical distribution of stopping times with sequential hypothesis test for different values of $r$. Here $\delta=0.05$ and $\epsilon=0.01$ is used.]{
		\includegraphics[height=4.3cm]{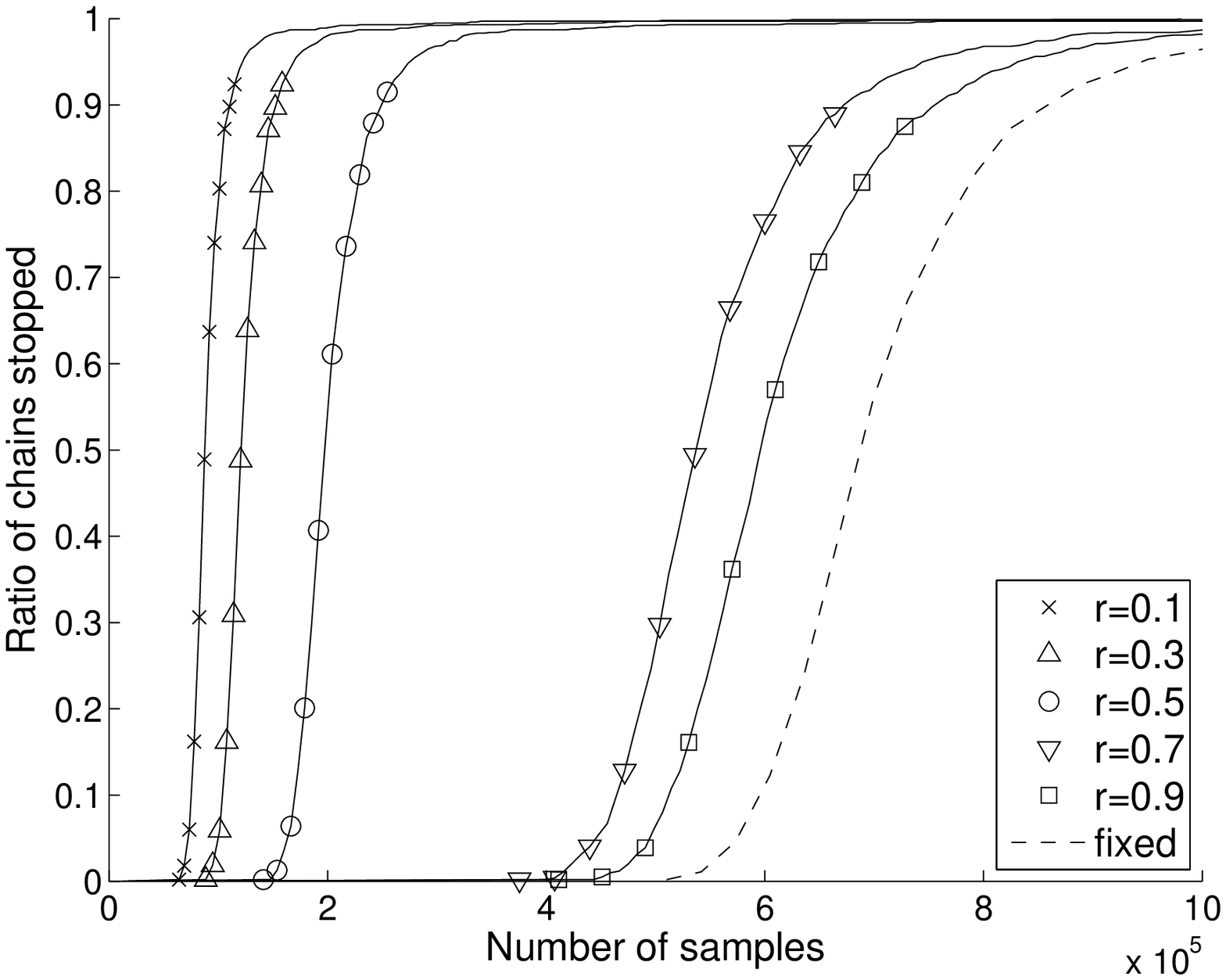}
		}
	&\\
	\addtolength{\subfigcapskip}{0.2cm}
	\subfigure[Mean empirical stopping times for sequential hypothesis test for different values of $\delta$, with $\epsilon=0.01$. Dashed lines show mean sample sizes required for the fixed sample size test.]{
		\includegraphics[height=4.3cm]{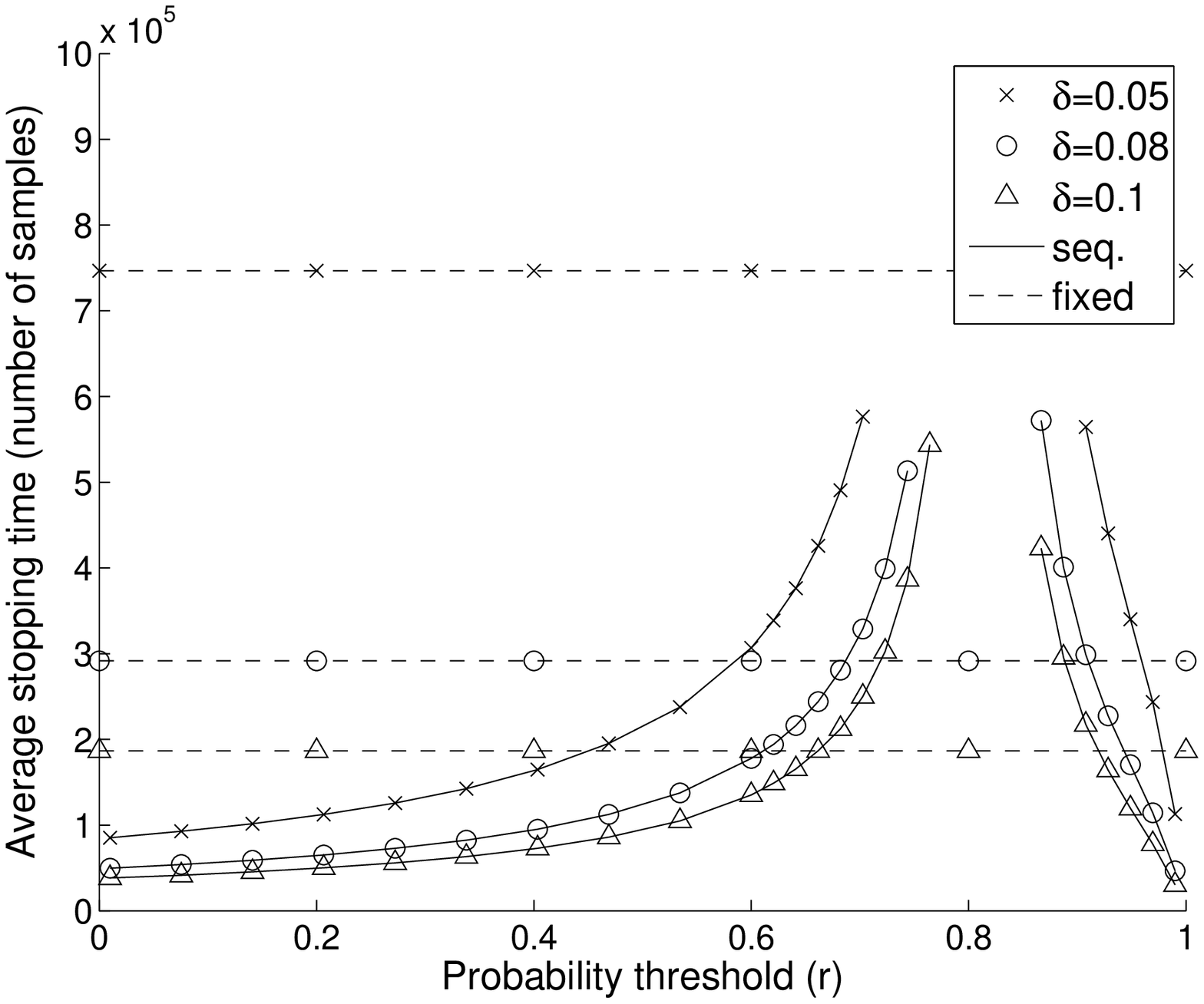}
		}
	&
	\addtolength{\subfigcapskip}{0.2cm}
	\subfigure[Mean empirical stopping times for sequential hypothesis test for different values of $\epsilon$, with $\delta=0.05$. Dashed lines show mean sample sizes required for the fixed sample size test.]{
		\includegraphics[height=4.3cm]{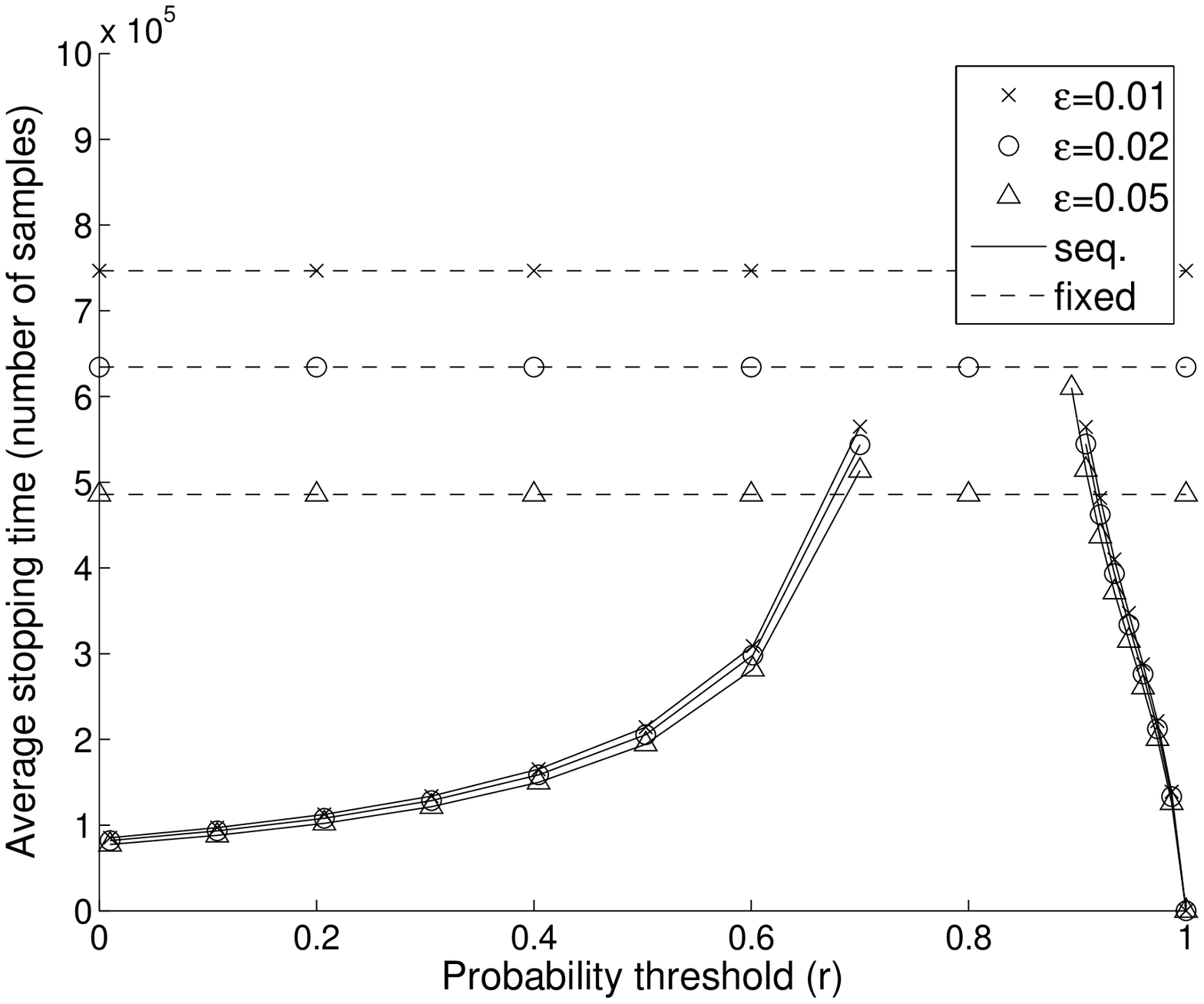}
		}
		
	\end{tabular}
	\caption{Results for the verification of $\psi_1$. }
	\label{fig:results_psi1}
\end{figure}

\subsection{Further properties}

We now look at two further properties regarding $\mathrm{STATn}$. Recall that property $\psi_1$ specified the behavior of $\mathrm{STATn}$ under transient Epo stimulation. Here we specify the behavior of STATn under two rounds of transient Epo stimulation ($\psi_2$) and under sustained Epo stimulation ($\psi_3$), (again, Epo is set externally and thus does not appear in the formulas). 

\noindent \textbf{Property 2} \emph{$\mathrm{STATn}$ reaches a \emph{high} level and then settles at a \emph{medium} level under two rounds of transient Epo stimulation }
\begin{equation}
\psi_{2} = F^{\le 60}([1 \le \mathrm{STATn} \le 2] \wedge F^{\le 60}(G^{\le 60}([0.5 \le \mathrm{STATn} \le 1]))).
\end{equation}

\noindent \textbf{Property 3} \emph{$\mathrm{STATn}$ reaches a \emph{very high} level and then settles at a \emph{very high} level under sustained Epo stimulation }
\begin{equation}
\psi_{3} = F^{\le 60}(G^{\le 60}([1.5 \le \mathrm{STATn} \le 2])).
\end{equation}

Table \ref{tab:results} summarizes the results of the verification with properties $\psi_1$ to $\psi_3$. (Run times were measured on a $2.83$ GHz computer with 8GB of RAM).

\begin{table}[h]
\small
\begin{center}
\begin{tabular}{|c|c|c|c|c|c|c|}
\hline
Property & $r$ & $\delta$ & $\epsilon$ & Outcome & Samples/time taken~(seq.) & Samples/time taken~(fixed)\\ \hline
$\PP_{\ge r}(\psi_1)$ & $0.7$ & $0.05$ & $0.01$ & True & $5.65\cdot 10^5$/1128s & $7.46\cdot 10^5$/1492s \\
$\PP_{\ge r}(\psi_2)$ & $0.8$ & $0.05$ & $0.01$ & True & $3.12\cdot 10^5$/618s & $7.46\cdot 10^5$/1477s\\
$\PP_{\ge r}(\psi_3)$ & $0.8$ & $0.05$ & $0.01$ & False & $7.80\cdot 10^4$/154s & $7.46\cdot 10^5$/1504s\\
\hline
\end{tabular}
\end{center}
\caption{Verification results on properties of the JAK-STAT pathway model. All numbers shown are mean values across $1000$ independent runs.}
\label{tab:results}
\end{table}

\section{Conclusion}
In this paper we proposed a method for performing probabilistic verification of a system conditioned on noisy observations, using dependent realizations of it's dynamics.
There are several directions along which results presented here can be generalized. Here we considered dynamical systems described as systems of ODEs. It is possible to generalize the methodology to continuous-time Markov chain (CTMC) and stochastic differential equation (SDE) models \cite{kwiatkowska2007stochastic,wilkinson2012stochastic}. In these models the system state over time is described by a stochastic process. By conditioning on observations, one can consider the posterior distribution of the system state (and any unobserved parameters), and verify the system's behavior with respect to this distribution. MCMC methods have been proposed for sampling the posterior in such models \cite{golightly2011bayesian}, and our probabilistic verification methods could be adapted to this context. 

In our case study, we verified properties on a model that has $4$ unknown parameters. The use of MCMC methods for sampling posterior distributions on considerably larger biochemical pathway models has been demonstrated under realistic conditions (see for instance \cite{xu2010inferring,eydgahi2013properties}). This suggests that our proposed verification procedure will also be applicable to larger problems. It may also be interesting to examine the use of methods other than MCMC, such as sequential Monte Carlo \cite{kantas2009overview} or approximate Bayesian computation \cite{toni2009approximate}, however, rigorous bounds on the required sample size in these settings is still an open question. 

In this work we posed probabilistic verification as a hypothesis testing problem. In a Bayesian model checking approach \cite{clarkebayes,clarkebayes2}, it is assumed that the probability of the satisfaction of a property is a random variable. One advantage of the Bayesian approach is that if useful priors are provided, the verification can be accomplished with significantly reduced sample size. It is conceptually straightforward to adapt our method and the sample size bounds to a Bayesian model checking setting and it is a possible future direction to pursue.

\bibliographystyle{splncs}
\bibliography{References}

\clearpage
\section{Appendix}\label{sec:apppathway}
Here we provide additional details on the JAK-STAT pathway model case study. The ODE equations governing the model are shown in Figure \ref{fig:jakstat}. The species in the model are as follows:

\vspace{10pt}
\begin{tabular}{|l|l|l|}
 \hline
 Name & Description & Init. amount \\ \hline
Epo & Erythropoietin, input stimulus & 2.0 \\ \hline
STAT & Unphosphorylates STAT monomer in cytoplasm & 0\\ \hline
STATp & Phosphorylated STAT monomer in cytoplasm & 0\\ \hline
STATpd & Phosphorylated STAT dimer in cytoplasm & 0\\ \hline
STATn & Total STAT in nucleus & 0\\ \hline
$X_1 \ldots X_K$ & Represent delay in STAT exiting nucleus (we use $K=10$) & 0\\ \hline
\end{tabular}

\begin{figure}
	\centering
	\begin{align*}
		\frac{d[\mathrm{STAT}]}{dt} & = -k_1[\mathrm{STAT}][\mathrm{Epo}] + 2k_4[X_K]\\
		\frac{d[\mathrm{STATp}]}{dt} & = k_1[\mathrm{STAT}][\mathrm{Epo}] - k_2[\mathrm{STATp}]^2\\
		\frac{d[\mathrm{STATpd}]}{dt} & = -k_3[\mathrm{STATpd}] + 0.5k_2[\mathrm{STATp}]^2\\
		\frac{d[\mathrm{X_1}]}{dt} & = k_3[\mathrm{STATpd}] -k_4[\mathrm{X}_1]\\
		\frac{d[\mathrm{X_j}]}{dt} & = k_4[\mathrm{X}_{i-1}]-k_4[\mathrm{X}_i] \quad , \quad \quad j=2\ldots K \\
		\frac{d[\mathrm{STATn}]}{dt} & = k_3[\mathrm{STATpd}]-k_4[\mathrm{X}_K]
	\end{align*}
\caption{ODE model of the JAK-STAT pathway under Epo stimulation.}
\label{fig:jakstat}
\end{figure}

Figure \ref{fig:expdata} shows the experimental data $Y$ used to define the posterior distribution. The data points, as well as standard deviations are obtained from experiments published in~\cite{swameye2003}. Figure \ref{fig:epoinput} shows $3$ different time courses for the externally set Epo stimulation used when verifying with respect to $\psi_1$, $\psi_2$ and $\psi_3$ respectively. 

\begin{figure}[H]
\centering
\includegraphics[width=10cm]{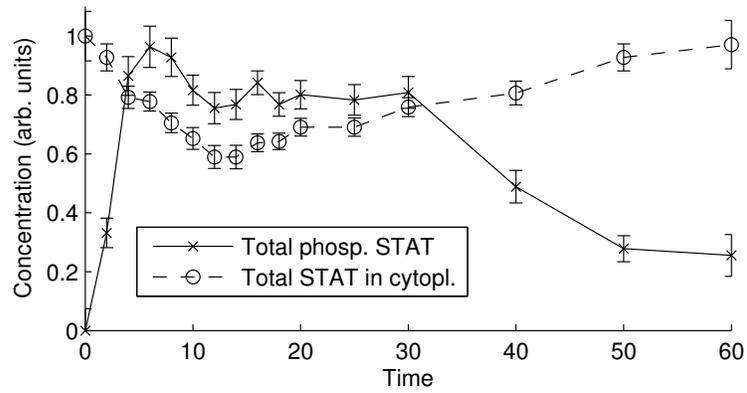}
\caption{Experimental data used for case study~\cite{swameye2003}. The Gaussian likelihood of parameters is evaluated using the shown data points and deviations.}
\label{fig:expdata}
\end{figure}

\begin{figure}[H]
\centering
\includegraphics[width=10cm]{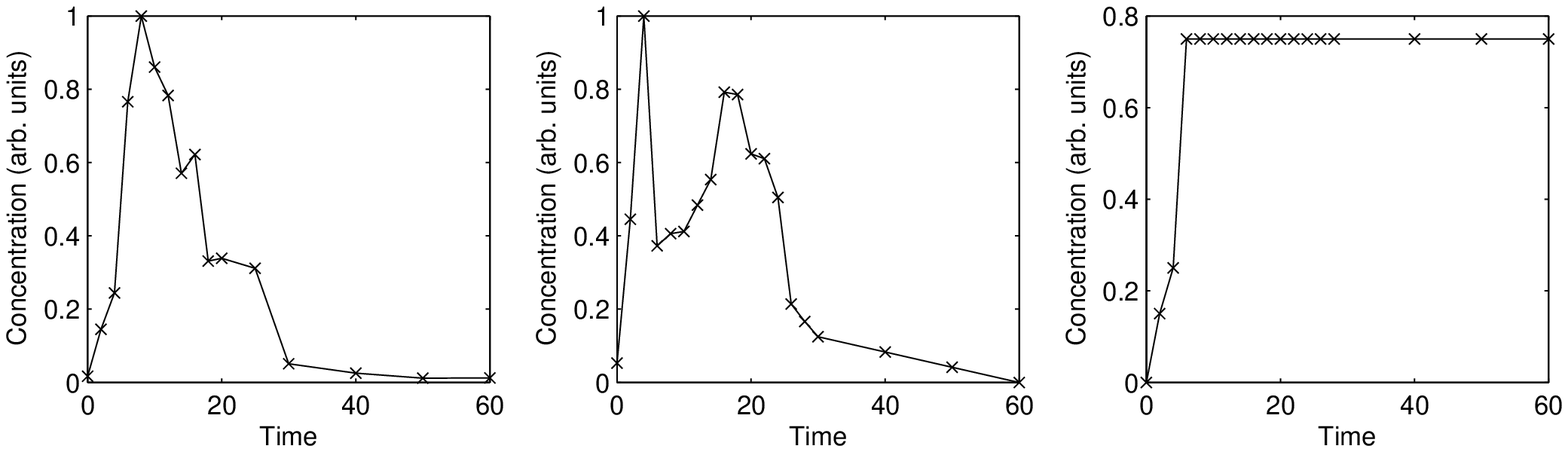}
\caption{Epo stimulation dynamics.  These time courses are used as (deterministic, externally fixed) inputs when verifying $\psi_1$, $\psi_2$ and $\psi_3$ respectively. Transient stimulation~\cite{swameye2003} (left), two rounds of transient stimulation~\cite{swameye2003} (center), sustained stimulation (right).}
\label{fig:epoinput}
\end{figure}

\end{document}